\documentclass[aps,prl,twocolumn,superscriptaddress]{revtex4}

\usepackage{amsfonts,epsfig}
\usepackage{latexsym}
\usepackage{epsfig}
\usepackage{amsmath}
\usepackage{amssymb}
\usepackage{amsfonts}
\usepackage{amsthm}
\usepackage{mathrsfs}
\usepackage{natbib}
\usepackage{color,verbatim,graphics}
\usepackage{psfrag}
\DeclareMathAlphabet{\mathrsfs}{U}{rsfs}{m}{n}
\DeclareMathAlphabet{\mathpzc}{OT1}{pzc}{m}{it}
\DeclareMathAlphabet{\matheus}{U}{eus}{m}{n}
\DeclareMathAlphabet{\mathbbold}{U}{bbold}{m}{n}

\newcommand{\ket}[1]{\left | #1 \right \rangle}
\newcommand{\bra}[1]{\left \langle #1   \right |}
\newcommand{\proj}[1]{\ket{#1}\bra{#1}}

\newcommand{\HH}{\mathcal{H}}

\newcommand{\Q}{\mathcal{Q}}

\newcommand{\de}[1]{\left ( #1 \right )}

\newcommand{\DE}[1]{\left \{ #1 \right \}}
\newcommand{\tr}[1]{\textrm{Tr}\de{#1}}

\newcommand{\ba}{\begin{eqnarray}}
\newcommand{\ea}{\end{eqnarray}}
\newcommand{\be}{\begin{equation}}
\newcommand{\ee}{\end{equation}}

\newcommand{\eg}{{\it{e.g.}}}
\newcommand{\ie}{{\it{i.e.}}}
\newcommand{\etal}{{\it{et al.}}}

\newtheorem{theorem}{Theorem}
\newtheorem{lemma}{Lemma}

\begin{document}

\title{Device-independent bounds for Hardy's experiment}

\author{Rafael Rabelo}
\affiliation{Centre for Quantum Technologies, National University of Singapore, 3 Science Drive 2, Singapore 117543}
\author{Law Yun Zhi}
\affiliation{Department of Physics, National University of Singapore, 2 Science Drive 3, Singapore 117542}
\author{Valerio Scarani}
\affiliation{Centre for Quantum Technologies, National University of Singapore, 3 Science Drive 2, Singapore 117543}
\affiliation{Department of Physics, National University of Singapore, 2 Science Drive 3, Singapore 117542}

\date{\today}


\begin{abstract}
In this Letter we compute an analogue of Tsirelson's bound for Hardy's test of nonlocality, that is, the maximum violation of locality constraints allowed by the quantum formalism, irrespective of the dimension of the system. The value is found to be the same as the one achievable already with two-qubit systems, and we show that only a very specific class of states can lead to such maximal value, thus highlighting Hardy's test as a device-independent self-test protocol for such states. By considering realistic constraints in Hardy's test, we also compute device-independent upper bounds on this violation and show that these bounds are saturated by two-qubit systems, thus showing that there is no advantage in using higher-dimensional systems in experimental implementations of such test.
\end{abstract}


\maketitle

\textit{Introduction.--} The development of quantum information science is based on a recurrent pattern: non-classical features of quantum physics, previously considered as mind-boggling and worth only of philosophical chat, are found to have an operational meaning and even to be potentially useful for applications. One of the discoveries that triggered this development is the prediction and observation of the violation of Bell inequalities \cite{Bell1964}. This observation implies that correlations obtained by measuring separated quantum systems locally cannot be simulated classically without communication, a fact that is often referred to as \textit{nonlocality}.

Within quantum information, nonlocality has undergone an interesting parable. For many years, it has been put aside as having fulfilled its role: the loathed local variables models having been disposed of forever, one could peacefully concentrate on entanglement theory. Only few researchers kept on believing that this very intriguing observation could be useful for something in itself. The latter view was vindicated a few years ago, when it was noticed that nonlocality allows \textit{device-independent} assessments: indeed, nonlocality is assessed only from the input-output statistics of the measurement, without reference to the degree of freedom that is being measured. This powerful type of assessment is sensitive to the existence of undesired side-channels and will be ideal for certification of future quantum devices. So far, device-independent results are available for the security of quantum cryptography \cite{Acin2007,Hanggi2010}, the quality of sources \cite{Bardyn2009,Mckague2012} and measurement devices \cite{Rabelo2011}, the amount of randomness that one can generate \cite{Pironio2010,Colbeck2011}. In this paper, we study the possibility of device-independent assessment of one of the earliest proposals to check nonlocality: it used to be called \emph{Hardy's paradox} but, in the spirit of quantum information, we'd rather call it \textit{Hardy's test} \cite{Hardy1992}


Hardy's test was originally stated by means of a particular experimental setup consisting of two overlapping Mach-Zehnder interferometers, one for electrons and one for positrons, arranged so that if the positron and the electron each take a particular path they will meet and annihilate one another. A paradox arises under the assumption of local realism: in any classical local theory a certain detection pattern must never occur, while quantum theory assign to its occurrence a nonzero probability, hereafter referred to as \textit{Hardy's probability}. It was soon realized that the argument could be extended to different states and measurements \cite{Hardy1993, Goldstein1994}, and proved to hold for almost all entangled pure states of two qubits, with maximum Hardy's probability equal to \ba
p_{\textrm{Hardy}}&=&\de{5\sqrt{5}-11}/2\approx 9\%\,.\label{phardy}\ea Interestingly, though, the maximaly entangled state of two qubits does not show nonlocality in Hardy's test. 

Hardy's test has been the object of several theoretical generalizations \cite{Clifton1992, Kunkri2005, Seshadreesan2011,Hillery2001,Pagonis1992, Kar1997a, Cereceda2004, Boschi1997, Liang2011} and has been implemented in experiments using photonic systems \cite{Torgerson1995, Giuseppe1997, Irvine2005, Fedrizzi2011, Valone2011}. The latter, however, had to consider deviations from the original proposal, where the probabilities of a set of observations - hereafter referred to as \textit{constraint probabilities} - were assumed to be strictly equal to zero, an obviously unrealistic requirement. One way to overcome this problem is to consider a nonideal version of Hardy's test, and to compute local bounds on Hardy's probability in terms of relaxed bounds on the constraint probabilities. The computed local bound, which a successful experiment must violate, turns out \cite{Braun2008, Mermin1994, Garuccio1995} to be equivalent to the Clauser-Horne (CH) Bell inequality \cite{CH1974}.

In this paper, we provide three device-independent results on Hardy's test. First, we consider the original (or \textit{ideal}) Hardy's test and prove that that \eqref{phardy} is the maximum value of $p_{\textrm{Hardy}}$ allowed by quantum physics, irrespective of the dimension; this is the analog of the Tsirelson bound \cite{tsi80}. A remarkable consequence of our derivation constitutes our second main result: any state that achieves \eqref{phardy} in the ideal test is equivalent, up to local isometries, to the unique two-qubit state that achieves that violation. This is a case of \textit{self-testing} \cite{Mayers04,Mckague2012}, the first that detects a non-maximally entangled state (see parallel work \cite{Yang12}). Finally, our third result is a proof that, even for \textit{nonideal} versions of Hardy's test, there is no practical advantage in using higher-dimensional systems.

\begin{figure}
	\centering
		$\psfrag{a}[][][1]{$a$}
		$\psfrag{b}[][][1]{$b$}
		$\psfrag{x}[][][1]{$x$}
		$\psfrag{y}[][][1]{$y$}
		\includegraphics[width = 0.35\textwidth]{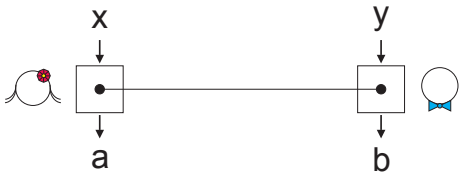}
	\caption{Schematic diagram for the Hardy's test scenario.}
	\label{fig:1}
\end{figure}

\textit{Hardy's test.--} Let us briefly summarize Hardy's test. Consider two parties, say, Alice and Bob, each of which is able to perform two possible measurements, $x = \DE{A_{0}, A_{1}}$ and $y = \DE{B_{0}, B_{1}}$, respectively, on its part of a shared physical system. Each measurement has two mutually exclusive outcomes, labeled by $a = \DE{\pm 1}$, for the measurements of Alice, and $b = \DE{\pm 1}$, for the ones of Bob. The situation considered by Hardy assumes the three constraint probabilities
\begin{subequations} \label{constraints}
\begin{align} 
p(+,+|A_0,B_0) & = 0, \\
p(+,-|A_1,B_0) & = 0, \\
p(-,+|A_0,B_1) & = 0.
\end{align}
\end{subequations}
Suppose there are measurement devices and physical systems such that these three equations are fulfilled. If this setup can be described by a local realistic theory, then it follows that
\begin{align} \label{hardy}
p_{\textrm{Hardy}} & \equiv p(+,+|A_1,B_1) = 0.
\end{align}
Hardy realized that in quantum mechanics there are measurements and a particular state of a two-qubit system such that the constraint probabilities are fulfilled while Hardy's probability is nonzero, leading to a so-called `paradox'. Extending the analisys to all possible measurements and states, Hardy later showed that the maximum value of $p_{\textrm{Hardy}}$ for systems of two qubits is $\de{5\sqrt{5}-11}/2$. A brute force calculation proved that this value cannot be exceeded using two three-dimensional systems \cite{Seshadreesan2011}. Here, we prove that this value is device-independent, that is, it is optimal for bipartite quantum systems of any dimension.

\begin{theorem}
The maximum value of Hardy's probability for quantum systems of arbitrary finite dimension is $p_{\textrm{Hardy}} = \de{5\sqrt{5}-11}/2$, just as for qubits.\end{theorem}

\begin{proof} In quantum mechanics, joint probabilities for the outcomes of measurements performed on space-like separated parts of a quantum system are given by
\begin{equation}
\label{born}
p(a,b|x,y) = \tr{\rho \Pi_{a|x}\otimes\Pi_{b|y}},
\end{equation}
where $\rho$ is the state of the system and $\Pi_{a|x}, \Pi_{b|y}$ are the measurement operators associated to outcomes $a, b$ of measurements $x, y$, respectively. The latter operators are POVM effects, in general; however, since we do not set any constraint on the dimension of the Hilbert space, Neumark's theorem allows us to consider only projective measurements, without loss of generality. The core of the proof exploits the following lemma, proven in \cite{Masanes2006}:

\begin{lemma}Given \textit{two} Hermitian operators $A_0$ and $A_1$ with eigenvalues $\pm 1$ acting on a Hilbert space $\HH$, there is a decomposition of $\HH$ as a direct sum of subspaces $\HH^{i}$ of dimension $d \le 2$ each, such that both $A_0$ and $A_1$ act within each $\HH^{i}$, that is, they can be written as $A_0 = \bigoplus_{i} A_0^{i}$ and $A_1 = \bigoplus_{i} A_1^{i}$, where $A_0^{i}$ and $A_1^{i}$ act on $\HH^{i}$.\end{lemma}

Let then $A_0 = \Pi_{+|A_0} - \Pi_{-|A_0}$ and $A_1 = \Pi_{+|A_1} - \Pi_{-|A_1}$, where $\Pi_{a|x}$ are projection operators. It follows from Lemma 1 that $\Pi_{a|x} = \bigoplus_{i} \Pi_{a|x}^{i}$, where each $\Pi_{a|x}^{i}$ acts on $\HH^{i}$, for all $a$ and $x$; we also denote $\Pi^i=\Pi_{+1|x}^{i}+\Pi_{-1|x}^{i}$ the projector on $\HH^{i}$. Needless to say, Lemma 1 is also valid on Bob's side; we use analog notations for Bob's operators. With these notations,  
\begin{subequations}
\begin{align}
p(a,b|x,y)  &= \sum_{i,j} q_{ij} \tr{\rho_{ij} \Pi_{a|x}^{i} \otimes \Pi_{b|y}^{j}}\\ &\equiv \sum_{i,j} q_{ij} p_{ij}(a,b|x,y),
\end{align}
\end{subequations}
where $q_{ij} = \tr{\rho \Pi^{i} \otimes \Pi^{j}}$ and $\rho_{ij} = \de{\Pi^{i} \otimes \Pi^{j} \rho \Pi^{i} \otimes \Pi^{j}} / q_{ij}$ is, at most, a two-qubit state. Since $q_{ij} \ge 0$ for all $i,j$ and $\sum_{i,j} q_{ij} = 1$, the constraint probabilities \eqref{constraints} are satisfied for $p$ if and only if they are satisfied for each of the $p_{ij}$. But, then,
\begin{equation}\label{hardysum}
p(+,+|A_1,B_1) = \sum_{i,j} q_{ij} p_{ij}(+,+|A_1,B_1),
\end{equation}
is a convex sum of Hardy's probabilities in each two-qubit subspace \footnote{It is important to note that for the maximum value of \eqref{hardysum} to be reached it is necessary that, for all $i,j$ such that $q_{ij} \neq 0$, the dimension of both $\HH^i$ and $\HH^j$ be equal to 2. This implies that the effective dimension $d$ of the local Hilbert space $\HH$ of the system is even, and that $A_{0}, A_{1}, B_{0}, B_{1}$ have exactly $d/2$ positive and negative eigenvalues.}. As a convex sum, it is less or equal to the largest element in the combination, whose maximum value is known to be given by \eqref{phardy}. This concludes the proof \footnote{There is an alternative, simpler proof of the above theorem that consists, basically, in noticing that any probability distribution that maximizes Hardy's probability is an extremal point of the set of quantum probability distributions. According to \cite{Masanes2006}, every extremal point, in this scenario, can be obtained from projective measurements on two-qubit systems, thus proving the stated result. The reason we opted for presenting the extensive proof is that it leads to interesting insights about the states that lead to such maximal violation, as discussed below. This proof cannot be extended to the nonideal scenario we later consider due to the fact that, in that scenario, it is not clear wether the points reaching maximal Hardy's probability are extremal or not.}. \end{proof}

\textit{Hardy's test leads to self-testing. --} It follows from the previous proof that $p(+,+|A_1,B_1)$ reaches its maximal value if and only if $p_{ij}(+,+|A_1,B_1)$ is maximal for every $ij$ such that $q_{ij} \neq 0$. The following Lemma, proved in \cite{Hardy1993, Goldstein1994}, states that only a very specific class of two-qubit states can lead to this maximal value:

\begin{lemma}\label{lemma2} Consider Hardy's test implemented in a two-qubit system, and let $A_{0} = B_{0} = \proj{0} - \proj{1}$. The probability $p_{\textrm{Hardy}}$ reaches its maximal value if, and only if, the state of the system is
\begin{equation}
\ket{\phi} = a\de{\ket{01} + \ket{10}} + e^{i\theta}\sqrt{1-2a^2}\ket{11},\label{defphi}
\end{equation}
and the other two measurements are $A_{1} = B_{1} = \proj{+} - \proj{-}$ with
$\ket{+} = \frac{1}{\sqrt{1-a^2}}\de{\sqrt{1-2a^2}\ket{0} - e^{i\theta}a\ket{1}}$, $a = \sqrt{\de{3-\sqrt{5}}/2}$ and $\theta$ is arbitrary.\end{lemma}

In view of this, one can conjecture that, if the maximal value of $p_{\textrm{Hardy}}$ is observed, the state must somehow be a direct sum of copies of $\ket{\phi}$. We proceed to prove that this is indeed the case:

\begin{theorem} If $p_{\textrm{Hardy}} = (5\sqrt{5}-11)/2$ is observed in an ideal Hardy's test [\ie, together with \eqref{constraints}], then the state of the system is equivalent up to local isometries to $\ket{\sigma}_{AB} \otimes \ket{\phi}_{A'B'}$, where $\ket{\phi}$ is given in \eqref{defphi}
and $\ket{\sigma}$ is an arbitrary bipartite state. In other words, the ideal Hardy's test constitutes a self-testing of $\ket{\phi}$.\end{theorem}

\begin{proof}
Without loss of generality, let us choose the eigenbases of $A_{0}$ and $B_{0}$ as the computational bases: $\Pi_{+|A_{0}}^{i} = \proj{2i}$, $\Pi_{-|A_{0}}^{i} = \proj{2i+1}$, $\Pi_{+|B_{0}}^{j} = \proj{2j}$, $\Pi_{-|B_{0}}^{j} = \proj{2j+1}$. Then, by Lemma \ref{lemma2}, $p_{ij}(+,+|A_1,B_1) = \tr{\rho_{ij} \Pi_{+|A_{1}}^{i} \otimes \Pi_{+|B_{1}}^{j}} = \de{5\sqrt{5}-11}/2$ if and only if $\rho_{ij} = \proj{\phi_{ij}}$, where
\begin{multline}
\ket{\phi_{ij}} = a\de{\ket{2i,2j+1} + \ket{2i+1,2j}} +\\ e^{i\theta}\sqrt{1-2a^2}\ket{2i+1,2j+1},
\end{multline}
and $a = \sqrt{\de{3-\sqrt{5}}/2}$ and arbitrary $\theta$. This way, a state $\ket{\psi}$ can lead to a maximal value of $p_{\textrm{Hardy}}$ if, and only if, it is given by
\begin{equation}
\ket{\psi} = \bigoplus_{i,j}\sqrt{q_{ij}} \ket{\phi_{ij}}.
\end{equation}
The coefficients $q_{ij}$ are arbitrary probabilities that, by definition, are constrained to the form $q_{ij} = r_{i} s_{j}$, where $r_{i}, s_{j} \geq 0$, $\sum_{i} r_{i} = \sum_{j} s_{j} = 1$. The angle $\theta$ cannot depend on the indices $i,j$ because $\Pi_{+|A_{1}}^{i}$ is uniquely defined by $\theta$ (cf. Lemma \ref{lemma2}), and, by definition, is independent of $j$; the same reasoning can be applied to $\Pi_{+|B_{1}}^{j}$. Now, following \cite{Mckague2012}, we append local ancilla qubits prepared in the state $\ket{00}_{A'B'}$ and look for local isometries $\Phi_{A}$ and $\Phi_{B}$ such that
\begin{equation}\label{isometry}
\de{\Phi_{A} \otimes \Phi_{B}} \ket{\psi}_{AB}\ket{00}_{A'B'} = \ket{\sigma}_{AB}\ket{\phi}_{A'B'},
\end{equation}
where $\ket{\sigma}$ is a bipartite `junk' state. This can indeed be achieved for $\Phi_{A} = \Phi_{B} = \Phi$ defined by the map
\begin{subequations}
\begin{align}
\Phi \ket{2k,0}_{CC'} & \mapsto \ket{2k,0}_{CC'}, \\ \Phi \ket{2k+1,0}_{CC'} & \mapsto \ket{2k,1}_{CC'}, 
\end{align}
\end{subequations}
for both $C=A,B$. \end{proof}

Up to now, self-testing was known only for maximally entangled states (see, \eg, \cite{Mckague2012} and references therein). A parallel, independent work by Yang and Navascu\'es provides a very general approach to the self-testing of bipartite non-maximally entangled states \cite{Yang12}. Remarkably, though, our Hardy point is not detected by that test \footnote{At least not up to the ${\cal Q}_2$ step of the hierarchy (T.H. Yang, private communication, 17 July 2012)}.

\textit{Hardy's experiment with realistic constraints.--} Suppose now that the constraint probabilities \eqref{constraints} in Hardy's experiment are not exactly equal to zero. In this case, the local bound on Hardy's probability is no longer zero, either: in general, it is given by the following inequality \cite{Mermin1994, Garuccio1995}:
\begin{multline}
p(+,+|A_{1},B_{1}) \leq p(+,+|A_{0},B_{0}) + \\ p(+,-|A_{1},B_{0}) + p(-,+|A_{0},B_{1}).
\end{multline}
This inequality is a re-writing of the CH inequality \cite{CH1974}, which is not surprising, since the CH inequality is the only relevant criterion for nonlocality in a scenario with two parties, two inputs and two outcomes. In other words, as noticed in \cite{Braun2008}, Hardy's experiment turns out to be a study of the violation of the CH inequality under further constraints about the values of some probabilities.

Let us now set
\begin{subequations}
\label{constraints2}
\begin{align} 
p(+,+|A_0,B_0) & \leq \epsilon, \\
p(+,-|A_1,B_0) & \leq \epsilon, \\
p(-,+|A_0,B_1) & \leq \epsilon,
\end{align}
\end{subequations} for some $\epsilon\geq 0$ \footnote{Notice that, if no-signaling holds, then $p(+,+|A_0,B_0)=\epsilon$ implies $p(+|A_0)\geq \epsilon$ and $p(-,+|A_0,B_1)=\epsilon$ implies $p(-|A_0)\geq \epsilon$. Therefore $\epsilon\leq \frac{1}{2}$. But the region of interest is in fact $\epsilon\leq \frac{1}{3}$ as explained just below.}. The local bound on Hardy's probability becomes
\begin{equation} \label{hardy2}
p(+,+|A_1,B_1) \leq 3\epsilon.
\end{equation}
For $\epsilon\geq \frac{1}{3}$, the bound is trivial and quantum physics certainly cannot violate it; while for $0\leq \epsilon< \frac{1}{3}$, quantum physics may lead to a violation of the local bound. As before, we want to assess the maximal quantum violation in a device-independent scenario, \ie, without making any assumption on the Hilbert space dimension. The previously stated theorem cannot be extended, so we take a different approach: first, we use semi-definite programs to obtain an upper bound on Hardy's probability, using the method of Navascu\'es, Pironio and Ac\'{\i}n \cite{Navascues2008}; second, by considering two-qubit systems we obtain a value that is certainly achievable with quantum systems. By noticing that the values thus obtained coincide, we conclude that we have obtained the optimal value for Hardy's probability, and that this value can be reached with two qubit systems.

In detail: let $\Q$ be the set of quantum joint probability distributions, that is, vectors of probabilities of the form \eqref{born}, for all $a,b,x,y$. We compute a device-independent upper bound on Hardy's probability by optimizing it not over quantum probabilities in the set $\Q$ but over a larger set of probabilities that is computationally tractable --- as opposed to $\Q$, that still lacks a better characterization. This set is one of an infinite hierarchy of sets $\Q_{1} \supset \Q_{2} \supset \dots \supset \Q_{n} \supset \dots$, defined in terms of semi-definite programs \cite{Navascues2008, Doherty2008}, proven to converge to the quantum set, $\textrm{lim}_{n \rightarrow \infty} \Q_{n} = \Q$. For several values of $\epsilon$ in the interval $0 \leq \epsilon\leq 1/3$, we optimize Hardy's probability over the set $Q_{3}$, enforcing the constraints \eqref{constraints2}. The implementation was done in MATLAB using semi-definite programming \cite{Sedumi, Yalmip}. The results form the solid line in Fig.~\ref{fig:2}. For the lower bound, we consider the most general mixed states of two qubits and POVM elements acting on those. The maximal value of the Hardy's probability is estimated using constrained nonlinear optimization methods in MATLAB. These methods are not guaranteed to converge to global maxima, though, and are in fact rather sensitive to seed conditions; each point on the dotted line in Fig.~\ref{fig:2} is the maximum obtained over $10^{4}$ runs, with random initial seeds.

\begin{figure}
	\centering
		\includegraphics[width = 0.5\textwidth]{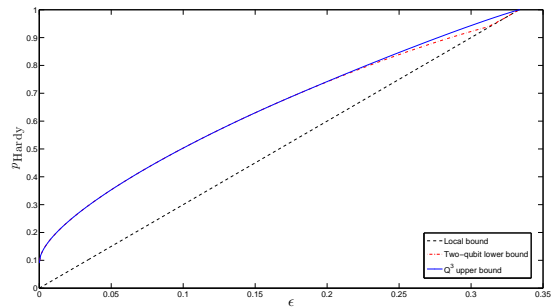}
	\caption{Upper and lower bounds on maximum Hardy's probability $p_{\textrm{Hardy}}$ in terms of the bound $\epsilon$ on the constraint probabilities. The solid (blue) line is the upper bound, computed from the set $\Q_{3}$; the dotted (red) line is the lower bound, computed from two-qubit systems; the dashed (black) line is the local bound.}
	\label{fig:2}
\end{figure}

The computed lower and upper bounds for Hardy's probability differ, at most, by values of order $10^{-2}$; in the region $\epsilon \lesssim 0.2$ (where any experiment that aims at implementing Hardy's test will have to be), this difference is of order $10^{-6}$. This proves that there is no advantage in using higher-dimensional systems, as compared to two-qubit systems, even in the presence of imperfections.

\textit{Conclusion.--} In this letter, we prove that the maximum value of Hardy's probability found for two-qubit systems, $\de{5\sqrt{5} -11}/2$, is the maximum one allowed by quantum theory, irrespective of the dimension of the system and of the measurements performed, that is, independend of the devices used. By showing that only a certain class of states can lead to such maximal value, we show that Hardy's test is, in fact, a self-testing protocol for such states. Extending the first results to a nonideal vesion of Hardy's test, where the constraint probabilities are no longer equal to zero, we compute device-independent upper bounds on Hardy's probability, in terms of the error parameter, and show that this bound is saturated by two-qubit systems.

Despite their fundamental importance, as the first proven analogue of Tsirelson's bound for Hardy's test, the results here presented also serve as a guideline for future experimental implementations, as they show that there is no advantage in using higher dimensional systems, as compared to two-qubit systems. 

\textit{Acknowledgements.--} The authors thank Tobias Fritz, Sibasish Ghosh and Marcelo Terra Cunha for discussions and comments. This work was supported by the National Research Foundation and the Ministry of Education, Singapore.


\end{document}